\documentclass[a4paper,11pt,titlepage]{llncs}

\usepackage{amssymb,amsmath}
\usepackage{graphicx}
\usepackage{fullpage}

\newcommand{\set}[1]{\ensuremath{\left\{#1\right\}}}
\newcommand{\setst}[2]{\left\{#1 \mid #2\right\}}
\newcommand{\zo}{\set{0, 1}}

\newcommand{\eps}{\ensuremath{\varepsilon}}
\renewcommand{\phi}{\ensuremath{\varphi}}

\newcommand{\Dc}{\ensuremath{\mathcal{D}}}
\newcommand{\Fc}{\ensuremath{\mathcal{F}}}

\newcommand{\Rbb}{\ensuremath{\mathbb{R}}}

\newcommand{\Prob}[2]{\ensuremath{\mathrm{Pr}_{#1}\left[#2\right]}}
\newcommand{\Exp}[2]{\ensuremath{\mathrm{E}_{#1}\left[#2\right]}}

\newcommand{\abs}[1]{\ensuremath{\left|#1\right|}}
\newcommand{\norm}[1]{\ensuremath{\left\| #1 \right\|}}

\newcommand{\supp}{\mathrm{supp}\;}
\newcommand{\KS}{C}
\newcommand{\Xcomment}[1]{}

\begin{document}

\title{Common information revisited}
\author{Ilya Razenshteyn\thanks{Supported in part by NAFIT ANR-08-EMER-008-01,
 RFBR 09-01-00709, and RFBR 10-01-93109 grants}}
\institute{Moscow State Lomonosov University, Mathematics Department,\\
Logic and Algorithms Theory Division, \texttt{ilyaraz@gmail.com}}
\maketitle

\begin{abstract}
One of the main notions of information theory is the notion of mutual information in two messages
(two random variables in Shannon information theory or two binary strings in algorithmic information theory).
The mutual information in $x$ and $y$ measures how much the transmission of $x$
can be simplified if both the sender and the recipient know $y$ in advance. 
G\'acs and K\"orner gave an example where mutual information
cannot be presented as common information (a third message easily extractable from both $x$ and $y$).
Then this question was studied in the framework of algorithmic information theory by An. Muchnik and A. Romashchenko
who found many other examples of this type. K. Makarychev and Yu. Makarychev found a new proof of G\'acs--K\"orner results
by means of conditionally independent random variables. 
The question about the difference between mutual and common information can be studied quantitatively: for a given $x$ and
$y$ we look for three messages $a$, $b$, $c$ such that $a$ and $c$
are enough to reconstruct $x$, while $b$ and $c$ are enough to reconstruct $y$.

In this paper:
\begin{itemize}
    \item We state and prove (using hypercontractivity of product spaces)
    a quantitative version of G\'acs--K\"orner theorem;
    \item We study the tradeoff between $\abs{a}, \abs{b}, \abs{c}$ for a random pair $(x, y)$ such that Hamming distance
    between $x$ and $y$ is $\eps n$ (our bounds are almost tight);
    \item We construct ``the worst possible'' distribution on $(x, y)$ in terms of the tradeoff between
    $\abs{a}, \abs{b}, \abs{c}$.
\end{itemize}
\end{abstract}
 
\section{Introduction}
\label{sec:introduction}

Let us start with a specific communication problem.
Say we have two $n$-bit strings $x$ and $y$ such that Hamming distance between them is exactly $\eps n$.
We want to deliver $x$ and $y$ to the corresponding network 
nodes (Fig.~\ref{network-1.mps}).
\begin{figure}[h]
 \begin{center}
     \includegraphics[scale=1]{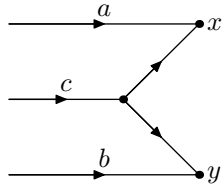}
 \end{center}
  \caption{Network information transmission problem}
  \label{network-1.mps}
\end{figure}
For that we send a common message $c$ to both nodes, and separate messages $a$ and $b$
to each node. What tradeoff between the lengths of $a$, $b$ and $c$ is possible (if we are willing to be able to
transmit any such pair $(x, y)$)?
If we require $c$ to be empty, then, clearly, $\abs{a}, \abs{b} \geq n$. On the other hand,
if $a$ and $b$ are empty, then $\abs{c} \geq (1 + H(\eps) + o(1)) n$.
In this paper we study this tradeoff and obtain almost tight bounds for $\abs{a}$, $\abs{b}$ and $\abs{c}$.

We can consider more general situation. Suppose we have two dependent random variables $X$ and $Y$. We sample
$n$ independent copies of $(X, Y)$ (let us call them $(X^1, Y^1)$, $(X^2, Y^2)$, \ldots, $(X^n, Y^n)$) and
want to transmit $x := X^1X^2\ldots X^n$ and $y := Y^1 Y^2 \ldots Y^n$. We want our transmission to be successful with
probability close to $1$.
Again we are interested in the tradeoff between
$\abs{a}$, $\abs{b}$ and $\abs{c}$.
Information Theory gives three trivial bounds:
\begin{eqnarray}
    \label{it-1}
    \abs{a} + \abs{c} \geq (H(X) + o(1))n,\\
    \label{it-2}
    \abs{b} + \abs{c} \geq (H(Y) + o(1))n,\\
    \abs{a} + \abs{b} + \abs{c} \geq (H(X, Y) + o(1)) n.
\end{eqnarray}
This transmission problem was first considered in \cite{GW74}.
In the paper~\cite{GK73} G\'acs and K\"orner showed that the first two bounds could not be tight simultaneously (that is,
``common information is less than mutual information'') unless
the joint distribution of $X$ and $Y$ is in some sense degenerate. 
Konstantin Makarychev and Yury Makarychev presented in~\cite{MM05} another proof of the same theorem using
a notion of conditional
independence of random variables.

One can also pose a similar problem in the framework of Algorithmic Information Theory.
Suppose we have two strings $x, y \in \zo^n$ such that $C(x) \approx n$, $C(y) \approx n$ and
$I(x : y) \approx 0.5 n$ (here $C(\cdot)$ stands for plain Kolmogorov complexity and $I(\cdot : \cdot)$ ---
for algorithmic mutual information).
Does there exist a string $z$ such that
\begin{itemize}
    \item $C(z) \leq 0.51 n$,
    \item $C(x \mid z) \leq 0.51 n$,
    \item $C(y \mid z) \leq 0.51 n$?
\end{itemize}

Again, it turns out that the answer is negative in general (see \cite{M86}, \cite{M98}, \cite{R00}, \cite{CMRSV02}).
Information-theoretic and Kolmogorov-complexity-theoretic formulations are similar: $a$ stands for the shortest program,
which converts $z$ to $x$, $b$ stands for the shortest program, which converts $z$ to $y$, and $c$ stands for the shortest
program, which outputs $z$.

All these problems can be reformulated using a simple unifying combinatorial notion of the profile of a bipartite graph.
Suppose we have a bipartite graph $E \subseteq X \times Y$. Consider the following communication problem: transmit endpoints
of an edge of $E$.

\begin{definition}
We say that a triple $(\alpha, \beta, \gamma)$ belongs to
the profile of $E$ if there exist two mappings $f \colon \zo^{\alpha} \times \zo^{\gamma} \to X$,
$g \colon \zo^{\beta} \times \zo^{\gamma} \to Y$ such that for every $(x, y) \in E$ there exist $a \in \zo^{\alpha}$,
$b \in \zo^{\beta}$, $c \in \zo^{\gamma}$ such that $x = f(a, c)$ and $y = g(b, c)$.
\end{definition}
There is an obvious one-to-one correspondence between the profile of $E$ and the set of possible lengths of three messages
that we may transmit if we are willing to be able to deliver endpoints of any edge in $E$.
We can consider profiles not only for graphs, but also for distributions (we require the existence of $a, b, c$
with high probability).
If we consider the adjacency matrix $M$ of $E$, then the abovementioned question can be reformulated as follows:
for which triples $(\alpha, \beta, \gamma)$
it is possible to cover all ones in $M$ with $2^{\gamma}$ combinatorial rectangles of size
$2^{\alpha} \times 2^{\beta}$?

In this paper we study profiles of various graphs and distributions.
In Section \ref{sec:comb-int}
we relate the profile of $E$ and the maximum number of ones in the rectangles in $M$ of a given size.
In one direction the relation is obvious, but suprisingly it turns out that in some cases (namely, for edge-transitive graphs)
the existence of a rectangle
with many ones implies the existence of a good rectangle cover.
We also prove several simple bounds on the profile of a bipartite graph which will be useful later.
In Section \ref{sec:hypercontractivity}
we devise a pretty general tool to upper-bound number of ones in rectangles in $M$ of a given size.
It relies on hypercontractivity in product spaces (in \cite{AG76} hypercontractivity was used for a similar information
problem).
Using it we prove a quantitative version of G\'acs--K\"orner Theorem (for the case of uniform marginal distributions)
and improve upon \cite{GK73} and \cite{MM05}.
In Sections \ref{sec:fixed-distance} and \ref{sec:gacs-revisited}
we prove almost tight bounds for
the profile of our first example (strings with Hamming distance $\eps n$) using all these tools.
Then in Section \ref{sec:stoch-min}
we build an example of a graph with a minimal possible profile (unsurprisingly, we just prove that
random graphs do the job).
This example was announced in \cite{CMRSV02} without proof.

\section{Profile and the maximal number of edges in a rectangle}
\label{sec:comb-int}

In this section we make two observations. Let $E\subseteq X\times Y$ be a set of edges of 
a bipartite graph. For given $\alpha$ and $\beta$ consider combinatorial rectangles $X'\times Y'$ where
$X'\subseteq X$ has cardinality $2^\alpha$ and $Y'\subseteq Y$ has cardinality $2^\beta$. Let
$R_E(\alpha,\beta)$ be the maximum number of edges covered by rectangle of this type.

\begin{proposition}\label{prop:upper}
  If $R_E(\alpha,\beta)\cdot 2^\gamma < \abs{E}$, then the triple $(\alpha,\beta,\gamma)$
  does not belong to the profile of $\abs{E}$.
\end{proposition}

This observation is obvious: each of $2^\gamma$ rectangles covers at most $R_E(\alpha,\beta)$ edges,
so they cannot cover the entire graph.

It turns out that this observation can be reversed for the case of edge-transitive graphs.
Recall that a bipartite graph $E\subseteq X\times Y$ is \emph{edge-transitive} if the group of its 
automorphisms acts transitively on $E$ (for every two edges $e$ and $e'$ there is an
automoprhism that maps $e$ to $e'$; automorphisms are pairs of permutations $X\to X$ and
$Y\to Y$ that generate a permutation of $E$).

\begin{proposition}\label{prop:lower}
   Assume that $E$ is edge-transitive. If $R_E(\alpha,\beta)\cdot 2^\gamma \ge \abs{E}$, then
the triple $(\alpha,\beta,\gamma+\log\log\abs{E})$ belongs to the profile. 
\end{proposition}

\begin{proof}
Let $R=X'\times Y'$ be a rectangle that contains many edges: let $K$ be the number of
edges in it, so that $K\cdot 2^\gamma\ge\abs{E}$. We will cover $E$ by shifted copies of $R$.
Consider $m$ independent randomly chosen automorphisms (all elements of the group are equiprobable);
let $R_1,\ldots,R_m$ be the images of $R$ under these automorphisms.

We want to show that $R_1,\ldots,R_m$ cover $E$ with positive probability. Indeed,
for a given edge $e$ the probability of being covered by one $R_i$ is $K/\abs{E}$
(the preimage of $e$ under the automoprhism is uniformly distributed in $E$ due to edge
transitivity). The probability of \emph{not} being covered is therefore $(1-K/\abs{E})$.
Different automorphisms are independent, so the probability of $e$ to avoid all $R_i$
is $(1-K/\abs{E})^m$. The probability that \emph{some} edge is not covered is bounded by
$\abs{E}\cdot (1-K/\abs{E})^m$. The assumption guarantees that $K/\abs{E}<2^{-\gamma}$, so $m=2^\gamma\log\abs{E}$
is enough to make this probability less than $1$. 
(Indeed, $(1-2^{-\gamma})^{2^\gamma}\approx 1/e<1/2$.) 

\end{proof}

In the typical application the values of $\alpha,\beta,\gamma$ are of the same order of
magnitude as $\log\abs{E}$, so $\log\log\abs{E}$ is small compared to $\alpha,\beta,\gamma$.

Now we state and prove several (trivial) bounds on the profile of $E$.

\begin{proposition}
\label{prop:trivial_bounds}
  Let $E \subseteq X \times Y$ be a bipartite graph without isolated nodes.
  Let $0 \leq \alpha \leq \log \abs{X}$ and $0 \leq \beta \leq \log \abs{Y}$.
  \begin{enumerate}
    \item
      If $\alpha + \gamma < \log \abs{X}$ or $\beta + \gamma < \log \abs{Y}$, then $\langle \alpha, \beta, \gamma\rangle$ does not belong to the profile of $E$.
    \item
      If $\alpha + \beta + \gamma < \log \abs{E}$, then $\langle \alpha, \beta, \gamma\rangle$ does not belong to the profile of $E$.
    \item
      If $\min(\alpha, \beta) + \gamma \geq \log \abs{E}$, then $\langle \alpha, \beta, \gamma\rangle$ belongs to the profile of $E$.
    \item
      If $\alpha + \beta + \gamma \geq \log \abs{X} + \log \abs{Y}$, then $\langle \alpha, \beta, \gamma\rangle$ belongs to the profile of $E$.
  \end{enumerate}
\end{proposition}
\begin{proof}
  \begin{enumerate}
    \item
      If $\alpha + \gamma < \log \abs{X}$, then we are unable to cover $X$ (here we use that there are no isolated nodes in $E$).
      The second case is similar.
    \item
      If $\alpha + \beta + \gamma < \log \abs{E}$, then it is obviously impossible to cover all edges.
    \item
      Using one $2^{\alpha} \times 2^{\beta}$ rectangle we can cover \emph{any} $2^{\min(\alpha, \beta)}$ 1's. Thus, we can cover $E$ with $\abs{E} / 2^{\min(\alpha, \beta)}$ rectangles.
    \item
      If $\alpha + \beta + \gamma \geq \log \abs{X} + \log \abs{Y}$, then we can cover not only $E$, but the entire
      $X \times Y$.
  \end{enumerate}
\end{proof}

\section{Upper-bounding $R_E(\alpha, \beta)$}
\label{sec:hypercontractivity}

To apply Proposition~\ref{prop:upper} to regular graphs, we need a technique to upper-bound $R_E(\alpha, \beta)$.
Let us consider slightly more general situation. Instead of a regular graph $E \subseteq X \times Y$
let us consider a distribution $\Dc$ over $X \times Y$ such that its marginal distributions are uniform.
Then the natural generalization of $R_E(\alpha, \beta) / \abs{E}$ is the following quantity:
$$
    R_{\Dc}(\alpha, \beta) := \max_{\begin{smallmatrix}\abs{A} \leq 2^{\alpha}\\\abs{B} \leq 2^{\beta}\end{smallmatrix}}
    \Prob{(x, y) \sim \Dc}{x \in A, y \in B}.
$$

Now let us generalize our problem even more. Suppose that $\Dc$'s marginal distributions are not necessarily uniform.
Let us denote them by $\Dc_X$ and $\Dc_Y$. Suppose we have two sets $A \subseteq X$, $B \subseteq Y$ such that
$\Prob{x \sim \Dc_X}{x \in A} \leq \mu$, $\Prob{y \sim \Dc_Y}{y \in B} \leq \nu$. We are interested in upper bounds on
$\Prob{(x, y) \sim \Dc}{x \in A, y \in B}$. There is an obvious bound $\min\set{\mu, \nu}$ for this quantity, but as
we will see if $\Dc$ is in some sense non-degenerate, then we can sharpen this bound.
From now on we assume that $\supp \; \Dc_X = X$ and $\supp \; \Dc_Y = Y$ (otherwise, we can reduce either $X$ or $Y$).

Let us call $\Dc$ \emph{non-degenerate} if $\supp \Dc$ is a connected bipartite graph on $(X, Y)$ and \emph{degenerate}
otherwise. 

To formulate the upper bound, we shall introduce a parameter $\delta(\Dc)$ with the following properties:
\begin{itemize}
    \item $\delta(\Dc) > 0$ iff $\Dc$ is non-degenerate (Theorem~\ref{delta_non_degenerate} below);
    \item $\delta(\Dc_1 \otimes \Dc_2) \geq \min \set{\delta(\Dc_1), \delta(\Dc_2)}$ (Theorem~\ref{tensor-product} below),
    here $\Dc_1 \otimes \Dc_2$ is a product distribution of $\Dc_1$ and $\Dc_2$;
    \item if $\delta(\Dc) > 0$, then there is a non-trivial upper bound on $\Prob{(x, y) \sim \Dc}{x \in A, y \in B}$
    (Theorem~\ref{hypercontractivity_rectangle_general} below). 
\end{itemize}

If $\Dc$'s matrix was symmetric, then one could in principle define $\delta(\Dc) := 1 - \lambda(\Dc)$,
where $\lambda(\Dc)$ is the second largest eigenvalue of $\Dc$'s matrix (the largest eigenvalue is clearly $1$).
Then, obviously, all three desired properties are true
(the third one is a corollary of Expander Mixing Lemma \cite{HLW06}). The problem is that Expander Mixing Lemma is too weak for our
purposes (especially if $\mu, \nu = o(1)$), so we need something stronger.

Now let us define $\delta(\Dc)$. Consider $\Fc_X$, $\Fc_Y$ --- linear spaces of $\Rbb$-valued functions on $X$ and $Y$
respectively. Consider the following linear operator $T_{\Dc} \colon \Fc_Y \to \Fc_X$:
$$
    (T_{\Dc}f)(x) := \Exp{(x', y) \sim \Dc}{f(y) \mid x = x'}.
$$
Let us consider standard $L_p$-norms on $\Fc_X$ and $\Fc_Y$: $\norm{f}_{p} := \Exp{}{\abs{f}^p}^{1/p}$ (here expectation is
taken over $X$ or $Y$ with respect to $\Dc_X$, $\Dc_Y$ respectively). It is easy to check
that $T_{\Dc}$ is an $L_p$-contraction for all $1 \leq p \leq \infty$.
\begin{lemma}
    \label{lp-contraction}
    $\norm{T_{\Dc} f}_p \leq \norm{f}_p$ for all $1 \leq p \leq \infty$ and $f \in \Fc_Y$.
\end{lemma}

$\delta(\Dc)$ characterizes to what extent Lemma~\ref{lp-contraction} can be sharpened for a particular $\Dc$
near $p = 2$.
\begin{definition}
    $$
        \delta(\Dc) := \max \setst{\delta \leq 1}
            {\forall f \in \Fc_Y \; \norm{T_{\Dc} f}_{2 + \frac{\delta}{1 - \delta}} \leq \norm{f}_{2 - \delta}}
    $$
\end{definition}
This definition makes sense because the following Lemma holds.
\begin{lemma}
    \label{norm-monotonicity}
    If $1 \leq p \leq q \leq \infty$, then
    $$
        \norm{f}_p \leq \norm{f}_q,
    $$
    there is an equality iff $f$ is constant.
\end{lemma}
Now let us restate the properties of $\delta(\Dc)$.
\begin{theorem}
    \label{delta_non_degenerate}
    $\delta(\Dc) > 0$ iff $\Dc$ is non-degenerate.
\end{theorem}
\begin{theorem}
    \label{tensor-product}
    $\delta(\Dc_1 \otimes \Dc_2) \geq \min \set{\delta(\Dc_1), \delta(\Dc_2)}$.
\end{theorem}
\begin{theorem}
    \label{hypercontractivity_rectangle_general}
    Let $A \subseteq X$, $B \subseteq Y$. If $\Prob{x \sim \Dc_X}{x \in A} \leq \mu$ and
    $\Prob{y \sim \Dc_Y}{y \in B} \leq \nu$, then
    \begin{multline*}
        \Prob{(x, y) \sim \Dc}{x \in A, y \in B} \leq \mu \nu + \\ +
        \left(\mu^{2 - \delta(\Dc)}(1 - \mu) + \mu (1 - \mu)^{2 - \delta(\Dc)}\right)^{1 / (2 - \delta(\Dc))}.
        \left(\nu^{2 - \delta(\Dc)}(1 - \nu) + \nu (1 - \nu)^{2 - \delta(\Dc)}\right)^{1 / (2 - \delta(\Dc))}.
    \end{multline*}
\end{theorem}
For the proof of Theorem~\ref{tensor-product} see \cite{GS11}. Theorem~\ref{delta_non_degenerate} and
Theorem~\ref{hypercontractivity_rectangle_general} will be proved in the Appendix.
We will typically apply Theorem~\ref{hypercontractivity_rectangle_general} in situations, where $\mu$ and $\nu$ tend to zero.
Let us instantiate Theorem~\ref{hypercontractivity_rectangle_general} for this case.
\begin{corollary}
    \label{hypercontractivity_rectangle}
    If $\mu = o(1)$ and $\nu = o(1)$ then
    $$
        \Prob{(x, y) \sim \Dc}{x \in A, y \in B} \leq O\left((\mu \nu)^{1 / (2 - \delta(\Dc))}\right).
    $$
\end{corollary}

Now we state and prove a quantitative version of G\'acs--K\"orner theorem~\cite{GK73}
for the case of uniform marginal distributions.

\begin{theorem}
    Let $\Dc$ be a distribution over $X \times Y$ with uniform marginal distributions.
    We sample $n$ independent copies $(X^1, Y^1)$, $(X^2, Y^2)$, \ldots, $(X^n, Y^n)$ of $\Dc$
    and want to transmit $x := X^1X^2\ldots X^n$ and $y := Y^1 Y^2 \ldots Y^n$ as in Fig.~\ref{network-1.mps}
    (with probability $1 - 2^{-\Omega(n)}$). Then,
    $$
        \abs{a} + \abs{b} + (2 - \delta(\Dc)) \abs{c} \geq (\log \abs{X} + \log \abs{Y} + o(1)) n.
    $$

\end{theorem}

That is, if $\Dc$ is non-degenerate, then the bounds (\ref{it-1}) and (\ref{it-2}) could not be tight simultaneously,
since $\delta(\Dc) > 0$, $H(X) = \log \abs{X}$, and $H(Y) = \log \abs{Y}$.

The proof is a trivial combination of Theorem~\ref{tensor-product}, Corollary~\ref{hypercontractivity_rectangle}, and
Proposition~\ref{prop:upper}.

\section{Fixed distance graph and its rectangles}
\label{sec:fixed-distance}

In this section we use Propositions~\ref{prop:upper} and~\ref{prop:lower} and the results of
Section~\ref{sec:hypercontractivity} to analyze the transmission of two strings with Hamming distance $\eps n$
(see the beginning of Section~\ref{sec:introduction}).

Let us consider the following bipartite graph $G_{n, \eps} \subseteq \zo^n \times \zo^n$. There is an edge $(x, y)$
iff Hamming distance between $x$ and $y$ is exactly $\eps n$.

We are interested in the profile of $G_{n, \eps}$. Let us denote $M_{n, \eps}$ its adjacency matrix.
For simplicity we restrict ourselves to the case where
$\alpha$ and $\beta$ (lengths of messages $a$ and $b$ sent to each node separately) are equal:
$\alpha = \beta = \tau n$, $\gamma = \varkappa n$, where $0 < \tau < 1$, $0 < \varkappa < 1 + H(\eps)$ are constants.
Let us denote by $\Lambda(\eps, \tau)$ the following fraction:
$$
  \Lambda(\eps, \tau) = \lim_{n \to \infty} \frac{\log (R(\eps, \tau, n) / \abs{G_{n, \eps}})}{n},
$$
where $R(\eps, \tau, n)$ is the maximum number of 1's in a rectangle of $M_{n, \eps}$
of size $2^{\tau n} \times 2^{\tau n}$, and $\abs{G_{n, \eps}}$ is the total number of edges in $G_{n, \eps}$.
Since $G_{n, \eps}$ is edge-transitive, bounds for $\Lambda$ can be directly translated into profile
bounds (via Propositions~\ref{prop:upper} and~\ref{prop:lower}).

Let us state our main combinatorial result:

\begin{theorem}
\label{fixed_distance_estimate}
    \textup(Lower bound\textup)
      If $\tau < H\left(1 - \sqrt{1 - \eps}\right)$, then $\Lambda(\eps, \tau) \geq -(1 + H(\eps) - 2\tau)$.
      If $\tau \geq H\left(1 - \sqrt{1 - \eps}\right)$, then 
      $$
        \Lambda(\eps, \tau) \geq -\left(1 + H(\eps) - \tau
                - \alpha H\left(\frac{\varepsilon}{2 \alpha}\right)
                - (1 - \alpha) H\left(\frac{\varepsilon}{2 (1 - \alpha)}\right)\right),
      $$
      where $0 < \alpha < 1/2$ and $H(\alpha) = \tau$.
    
    \textup(Upper bound\textup)
      $\Lambda(\eps, \tau) \leq -\frac{1 - \tau}{1 - \eps}$.
\end{theorem}

\subsection*{The lower bound on $\Lambda(\eps, \tau)$}

To prove that $\Lambda(\eps, \tau)$ is large
we show that some rectangle in $M_{n, \eps}$ has many $1$'s. Indeed, consider
a rectangle $C \times C$, where elements of $C$ 
are strings that contain exactly $\alpha n$ ones;
obviously there are $2^{(H(\alpha) + o(1)) n}$ of them (i.e. $\tau = H(\alpha) + o(1)$).
The number of 1's in $C \times C$ equals
$$
  \binom{n}{\alpha n} \binom{\alpha n}{\eps n / 2} \binom{(1 - \alpha) n}{\eps n / 2}
  = 2^{(H(\alpha) + \alpha H(\eps / 2 \alpha) + (1 - \alpha) H(\eps / 2 (1 - \alpha)) + o(1)) n}.
$$

The total number of 1's in $M_{n, \eps}$ is
$$
  2^{n} \binom{n}{\eps n} = 2^{(1 + H(\eps) + o(1)) n}.
$$

Thus
\begin{multline*}
  \Lambda(\eps, \tau) \geq \lim_{n \to \infty} \frac{1}{n} \cdot \log \frac{2^{(H(\alpha) + \alpha H(\eps / 2 \alpha) + (1 - \alpha) H(\eps / 2 (1 - \alpha)) + o(1)) n}}{2^{(1 + H(\eps) + o(1)) n}} =\\
      = -(1 + H(\eps) - H(\alpha) - \alpha H(\eps / 2 \alpha) - (1 - \alpha) H(\eps / 2 (1 - \alpha))).
\end{multline*}

If $\tau = H\left(1 - \sqrt{1 - \eps}\right)$, then the sphere gives a rectangle with $2^{(2\tau + o(1)) n}$ ones,
which is clearly optimal. For $\tau' < \tau$ we can subsample this rectangle and get $2^{\tau' n} \times 2^{\tau' n}$
rectangle with $2^{(2\tau' + o(1)) n}$ ones.


\subsection*{The upper bound on $\Lambda(\eps, \tau)$}

Let us prove that $\Lambda(\eps, \tau) \leq -\frac{1-\tau}{1-\eps}$.
Consider the distribution $D_{n, \eps}$ on $(x, y) \in \zo^n \times \zo^n$: 
$x$ is uniformly distributed in $\zo^n$; and $y\in\zo^n$ is obtained from $x$
by independently changing each bit with probability $\varepsilon$. 

This distribution generates an edge in $G_{n,\varepsilon}$ with probability at least $1/n$,
and all the edges of $G_{n,\varepsilon}$ are equiprobable. So instead of counting the number of edges in
a rectangle, we may estimate the $D_{n,\varepsilon}$-probability of this rectangle (the 
factor $n$ does not matter with our precision). It is enough to show, therefore, that
for every $C_1, C_2 \subseteq \zo^n$ such that
$\abs{C_1} = \abs{C_2} = 2^{\tau n}$ the following inequality holds:
\begin{equation}
  \label{continuous_isoperimetry}
  \Prob{(x, y) \sim D_{n, \eps}}{(x, y) \in C_1 \times C_2} \leq
  2^{-(\frac{1-\tau}{1-\eps}+o(1))n}.
\end{equation}

If we show that $\delta(D_{n,\varepsilon}) \geq 2 \varepsilon$, then we can plug this bound
into Corollary~\ref{hypercontractivity_rectangle}, and obtain (\ref{continuous_isoperimetry}).
Since
$D_{n, \varepsilon} = D_{1, \varepsilon}^{\otimes n}$, using Theorem~\ref{tensor-product} one can reduce this statement
to the following well-known inequality.
\begin{theorem}[Two-Point Inequality]
    $$
        \delta(D_{1, \varepsilon}) \geq 2 \varepsilon.
    $$
\end{theorem}

For the proof see \cite{GS11}.

\section{The profile of the fixed-distance graph}
\label{sec:gacs-revisited}

First, we use Theorem~\ref{fixed_distance_estimate} and Propositions~\ref{prop:lower} and~\ref{prop:upper} to
get explicit bounds for the combinatorial profile of the fixed-distance 
graph $G_{n, \eps}$ that improve those given by Proposition~\ref{prop:trivial_bounds}.

\begin{theorem}
  \label{fixed_profile_bounds}
  Let $0 < \tau, \varkappa < 1$ be constants.
  \begin{itemize}
    \item
	  If $\varkappa < (1 - \tau) / (1 - \eps)$, then for sufficiently large $n$ the triple
	  $\langle \tau n, \tau n, \varkappa n\rangle$ does not belong to the profile of $G_{n, \eps}$.
	\item
      There are two following ``positive bounds''.
      \begin{itemize}
        \item
            If $\tau < H\left(1 - \sqrt{1 - \eps}\right)$ and $\varkappa > 1 + H(\eps) - 2\tau$,
            then for sufficiently large $n$ the triple $\langle \tau n, \tau n, \varkappa n\rangle$
            belongs to the profile of $G_{n, \eps}$.
        \item
            If $\tau \geq H\left(1 - \sqrt{1 - \eps}\right)$ and
            $$
                \varkappa > 1 + H(\eps) - \tau
                        - \alpha H\left(\frac{\varepsilon}{2 \alpha}\right)
                        - (1 - \alpha) H\left(\frac{\varepsilon}{2 (1 - \alpha)}\right),
            $$
            where $0 < \alpha < 1/2$ and $H(\alpha) = \tau$,
            then for sufficiently large $n$ the triple $\langle \tau n, \tau n, \varkappa n\rangle$
            belongs to the profile of $G_{n, \eps}$.
      \end{itemize}
  \end{itemize}
\end{theorem}

Figure~\ref{fixed_distance_figure} shows the bounds for the profile for $\varepsilon=0.11\ldots$
(for this value the total number of edges is $2^{1.5n}$). It shows trivial upper and lower bounds from Proposition~$\ref{prop:trivial_bounds}$,
as well as our results (Theorem~\ref{fixed_distance_estimate}).

Not that our bounds are tight in two regions:
\begin{itemize}
    \item if $\tau < H(1 - \sqrt{1 - \eps})$, then our upper bound is equal to the trivial lower bound;
    \item if $\tau = 1 - o(1)$, then our upper bound is asymptotically equal to our lower bound.
\end{itemize}

\begin{figure}
\begin{center}  
\includegraphics[width = 0.6\textwidth]{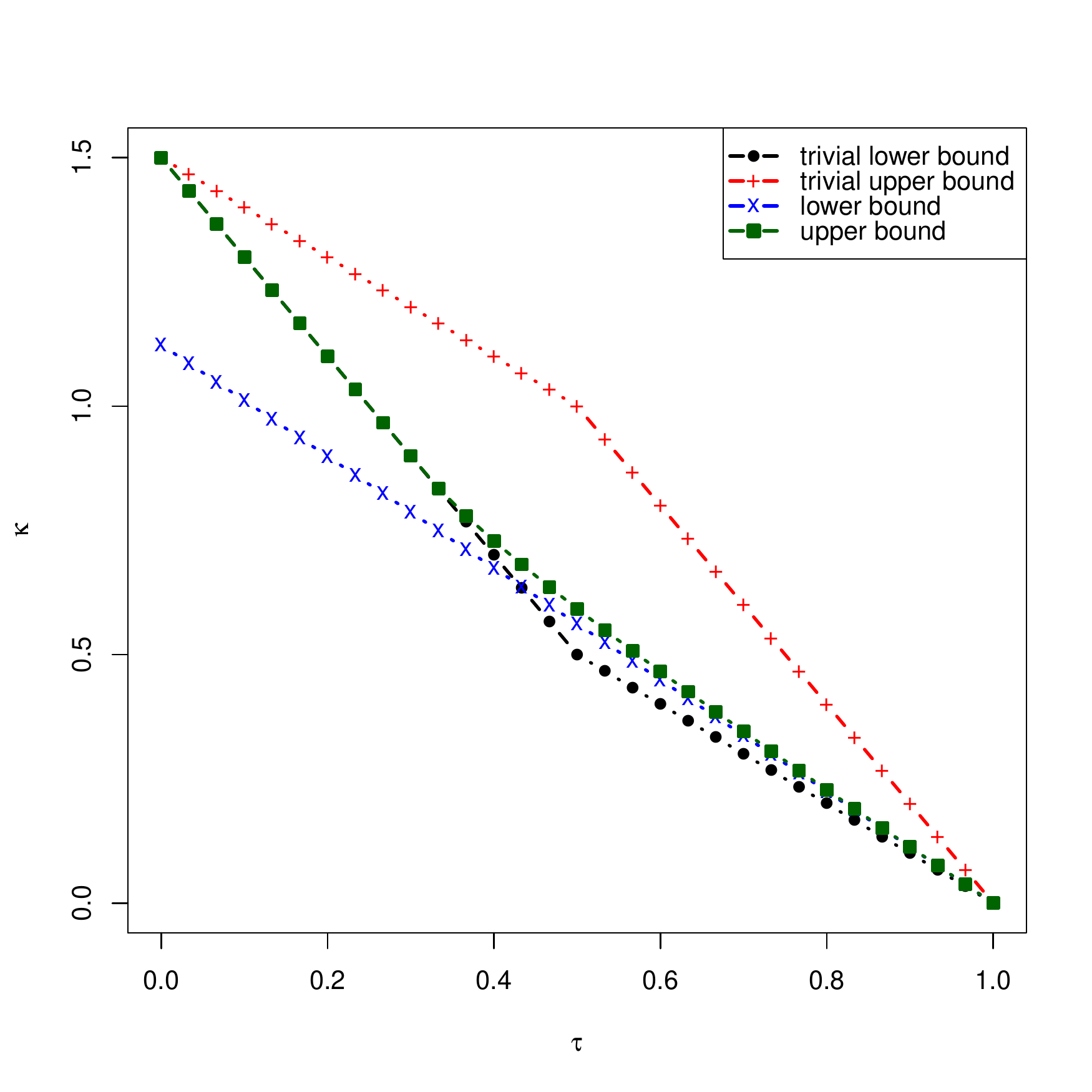}
\end{center}
  \caption{bounds on the profile from Proposition~\ref{prop:trivial_bounds} and Theorem~\ref{fixed_distance_estimate}, $d(x, y) = \eps n$, $\eps = 0.11\ldots$}
  \label{fixed_distance_figure}
\end{figure}

The same bounds can be obtained for other notions of profile, so the results are directly
comparable with previous work. Let us show how this can be done for Kolmogorov complexity.

Let us assume that for every $n$ a bipartite graph $E_n\subseteq\{0,1\}^n\times\{0,1\}^n$ is fixed
(and there is an algorithm computing $E_n$ given $n$). Let $\alpha,\beta,\gamma$ be some positive
rational numbers. Let $R_n(\alpha,\beta)$ be the maximum number of edges covered by a rectangle
$X'\times Y'$ with $\abs{X'}=2^{\alpha n}$ and $\abs{Y'}=2^{\beta n}$.

\begin{proposition}
\label{prop:kolmogorov-bounds}
    There exists some constant $d$ such that:

\begin{itemize}

 \item If $R_n(\alpha,\beta)\cdot 2^{\gamma n} \le \abs{E_n}$, then for all sufficiently large
$n$ for most edges $(x,y)\in E_n$ there is no string $c$ such that $$\KS(c)<\gamma n - d\log n, 
\quad \KS(x|c) < \alpha n - d \log n,
\quad \KS(y|c) < \beta n - d \log n.$$

 \item If $E_n$ is edge-transitive and $R_n(\alpha,\beta)\cdot 2^{\gamma n}\ge\abs{E_n}$,
then for all sufficiently large $n$ for every edge $(x,y)\in E_n$ there exists a string $c$
such that $$\KS(c)<\gamma n + d\log n, 
\quad \KS(x|c) < \alpha n +d\log n,
\quad \KS(y|c) < \beta n +d\log n.$$
\end{itemize}
\end{proposition}

\begin{proof}
 
For each $c$ we consider a rectangle $X'\times Y'$ where $X'$ consists of strings $x$ such that 
$\KS(x|c)<\alpha n$ and $Y'$ consists of strings $y$ such that $\KS(y|c)<\beta n$ ($\abs{X'} < 2^{\alpha n}$, $\abs{Y'} < 2^{\beta n}$).
The number of strings $c$ such that $\KS(c) < \gamma n - d \log n$ is less than $2^{\gamma n} / n^d$, so
the edges covered by these rectangles form a minority ($1/n^d$-fraction); for
all other pairs there is no string $c$ with described properties.

The second part of the theorem is also easy. Proposition~\ref{prop:lower} guarantees that
for sufficiently large $n$ the set $\abs{E_n}$ can be covered by $2^{\gamma n+ O(\log n)}$ rectangles
of size $2^{\alpha n}\times 2^{\beta n}$. By exhaustive search, we can find first covering
of this sort (in some natural order), and this covering is determined by $n$, so its
complexity is $O(\log n)$. Then we let $c$ be the number of the rectangle $X'\times Y'$ 
that covers a given edge $(x,y)$. The complexity of $c$ is at most $\gamma n + O(\log n)$. 
Knowing $c$ (and the
entire covering, which has complexity $O(\log n)$) we can describe $x$ and $y$ by their
numbers in $X'$ and $Y'$.

\end{proof}

One can also show that a random pair $(x,y)$ generated with distribution $D_{n,\varepsilon}$
will have a profile (in terms of complexity) within the bounds from Theorem~\ref{prop:kolmogorov-bounds} (with high probability). 
Indeed, the law of large number says that the number of places where $x$ and $y$ differ
is close to $\varepsilon n$, and for each fixed number of difference we get a uniformly
random edge in $G_{n,\varepsilon'}$ for $\varepsilon'$ close to $\varepsilon$. It remains
to note that our bounds are continuous (as functions of $\varepsilon$).

Let us note for comparison that (a weaker) upper bound for this distribution
can be obtained using conditional independence technique from~\cite{R00}.
It is quite involved, but it also has a form $\varkappa \geq c(\eps) (1 - \tau)$
(where $c(\cdot)$ is an explicit (but cumbersome) function)
as in Theorem~\ref{fixed_profile_bounds}.
So, let us compare $c(\eps)$ with $1/(1-\eps)$ (the bigger value is better) for different values of $\eps$ (see Fig.~\ref{comparison}).

\begin{figure}
\begin{center}
\begin{tabular}{|c|c|c|}
  \hline
  $\eps$ & $1/(1-\eps)$ & $c(\eps)$\\
  \hline
  0.1 & 1.11\ldots & 1.000015\ldots\\
  0.2 & 1.25 & 1.016\ldots\\
  0.3 & 1.43\ldots & 1.067\ldots\\
  0.4 & 1.67\ldots & 1.33\ldots\\
  \hline
\end{tabular}
\end{center}
\caption{A comparison of upper bounds from Theorem~\ref{fixed_distance_estimate} and from \cite{R00}}
\label{comparison}
\end{figure}

\section{A stochastic pair with minimal profile}
\label{sec:stoch-min}

A pair with minimal profile is constructed in~\cite{CMRSV02} using
the technique developed by An.~Muchnik~\cite{M86}, \cite{M98}. However, this construction
is quite artificial and cannot be translated into Shannon information theory, because
the constructed pair is not a typical object in a simple family (is not stochastic in
the sense of algorithmic information theory). In this section
we show how to construct a graph with minimal combinatorial profile (for graphs with given
number of edges). For that we first prove (using probabilistic arguments) that such
a graph exists; after that it can be found by brute-force search. This gives us a stochastic
pair with minimal profile. This result was announced in~\cite{CMRSV02} as
an unpublished result of An.~Muchnik (1958--2007) and was not published since then.

To analyze random graphs we need a version of Chernoff inequality
which deals with negatively correlated random variables~\cite{PS97}.

\begin{theorem}[Chernoff inequality]
      \label{chernov}
  Let $X_1, X_2, \ldots, X_n$ be negatively correlated binary random variables 
\textup(i.e., for every $i_1 < i_2 < \ldots < i_k$ 
we have $\Prob{}{X_{i_1} = 1 \wedge X_{i_2} = 1 \wedge \ldots \wedge X_{i_k} = 1} 
\leq \prod_{j=1}^k \Prob{}{X_{i_j} = 1}$\textup).

  Let $\mu = \Exp{}{X_1 + X_2 + \ldots + X_n}$ and $\delta > 0$.
  Then
  $$
    \Prob{}{X_1 + X_2 + \ldots + X_n \geq (1 + \delta) \mu} \leq \left(\frac{e^{\delta}}{(1 + \delta)^{1 + \delta}}\right)^{\mu}.
  $$
\end{theorem}

Following~\cite{CMRSV02}, we consider (for simplicity) graphs with $2^n$
vertices in each part and $2^{1.5n}$ edges. We choose a random graph of this type (i.e.,
uniformly choose a matrix of size $2^n\times 2^n$ with $2^{1.5n}$ ones).
 
\begin{theorem}
  \label{uniform_matrix}
  If $\varkappa + \tau < 1.5$ and $\varkappa + 2 \tau < 2$, then for some
  $\varepsilon > 0$ the probability of the event
  ``every $2^{\tau n} \times 2^{\tau n}$ rectangle in $M$
  contains less than $2^{(1.5 - \varkappa - \varepsilon) n}$ ones'' is close to~$1$.
\end{theorem}
\begin{proof}  
Let us choose $0 < \eps < \min(2-\varkappa-2\tau, 1.5-\varkappa-\tau)$.
The entries of $M$ are negatively correlated, so
the Chernoff inequality (with $1 + \delta = 2^{(2 - \varkappa - 2\tau - \eps) n}$ 
and $\mu = 2^{(2\tau-0.5)n}$) guarantees that for a fixed rectangle of 
size $2^{\tau n}\times 2^{\tau n}$ the probability of the event ``this rectangle 
contains more than $2^{(1.5 - \varkappa - \varepsilon) n}$ ones'' is bounded by
  \begin{equation*}
    \left(\frac{e^{2^{(2-2\tau-\varkappa-\eps)n} - 1}}
    {2^{(2-2\tau-\varkappa-\eps)n
    2^{(2-2\tau-\varkappa-\eps)n}}}\right)^{2^{(2\tau - 0.5) n}} =
    2^{-2^{(1.5 - \varkappa - \eps + o(1)) n}}
  \end{equation*}
Therefore the probability that some $2^{\tau n} \times 2^{\tau n}$ rectangle has too many ones does not exceed
  \begin{equation*}
    \binom{2^n}{2^{\tau n}}^2\cdot  2^{-2^{(1.5 - \varkappa - \eps + o(1)) n}}
   \leq
    2^{2^{(\tau + o(1)) n}-2^{(1.5 - \varkappa - \eps + o(1)) n}}.
  \end{equation*}
  Since $\eps < 1.5 - \varkappa - \tau$, we are done.
\end{proof}
\textbf{Remark.} Note that it follows from Propositions~\ref{prop:upper}
and~\ref{prop:trivial_bounds}
that the bounds in Theorem~\ref{uniform_matrix} are the best possible.

Another simple observation: with high probability all rows and columns of $M$
contain at most $2\cdot 2^{0.5n}$ elements. Indeed, the probability that in a
given row (or column) the number of ones is twice more than the expected value $2^{0.5n}$,
is doubly exponentially small (it follows from theorem~\ref{chernov}), and we have
only exponentially many rows and columns.

Then we can follow the plan described above: we conclude that there is a
graph with both properties and logarithmic complexity; then, 
using the first part of proposition~\ref{prop:kolmogorov-bounds} one can easily
see that a typical edge $(x,y)$ of this graph has minimal profile in the sense explained 
in~\cite{CMRSV02}. (The second property is needed to show that the complexities
of $x$ and $y$ in a typical edge are close to $n$.) 

\section{Acknowledgements}

I would like to thank Andrei Romashchenko and Alexander Shen for posing the problem
and for fruitful discussions, and Alex Samorodnitsky for useful advice.

\bibliography{ir}

\begin{thebibliography}{10}

\bibitem{GW74}
Robert Gray and Aaron Wyner.
\newblock {Source coding over simple networks}.
\newblock {\em Bell Systems Technical Journal}, 53(9):1681--1721, 1974.

\bibitem{GK73}
Peter G{\'{a}}cs and J{\'{a}}nos K{\"{o}}rner.
\newblock Common information is far less than mutual information.
\newblock {\em Problems of Control and Information Theory}, 2(2):149--162,
  1973.

\bibitem{MM05}
Konstantin Makarychev and Yury Makarychev.
\newblock {Conditionally independent random variables}.
\newblock Preprint, 2005.

\bibitem{M86}
Andrej Muchnik.
\newblock {On the extraction of common information of two words}.
\newblock {\em Pervyi vsemirnyi kongress obshchestva matematicheskoi statistiki
  i teorii veroyatnostei imeni Bernoulli. Tezisy}, 1:453, 1986.

\bibitem{M98}
Andrej Muchnik.
\newblock {On common information}.
\newblock {\em Theoretical Computer Science}, 207(2):319--328, 1998.

\bibitem{R00}
Andrei Romashchenko.
\newblock {Pairs of words with nonmaterializable mutual information}.
\newblock {\em Problems of Information Transmission}, 36:1--18, 2000.

\bibitem{CMRSV02}
Alexey Chernov, Andrey Muchnik, Andrei Romashchenko, Alexander Shen, and
  Nikolai Vereshchagin.
\newblock {Upper semi-lattice of binary strings with the relation ``x is simple
  conditional to y''}.
\newblock {\em Theoretical Computer Science}, 271:69--95, 2002.

\bibitem{AG76}
Rudolf Ahlswede and Peter G{\'{a}}cs.
\newblock {Spreading of Sets in Product Spaces and Hypercontraction of the
  Markov Operator}.
\newblock {\em The Annals of Probability}, 4(6):925--939, 1976.

\bibitem{HLW06}
Shlomo Hoory, Nathan Linial, and Avi Wigderson.
\newblock Expander graphs and their applications.
\newblock {\em Bulletin of the AMS}, 43:439--561, 2006.

\bibitem{GS11}
Christophe Garban and Jeffrey Steif.
\newblock {Lectures on noise sensitivity and percolation}.
\newblock Preprint, 2011.

\bibitem{PS97}
Alessandro Panconesi and Aravind Srinivasan.
\newblock {Randomized Distributed Edge Coloring via an Extension of the
  Chernoff-Hoeffding Bounds}.
\newblock {\em SIAM Journal on Computing}, 26(2):350--368, 1997.

\end{thebibliography}
\bibliographystyle{unsrt}

\appendix

\section{Properties of $\delta(\Dc)$}

We prove properties of $\delta(\Dc)$ that were promised in Section~\ref{sec:hypercontractivity}.

\begin{theorem}
    $\delta(\Dc) > 0$ iff $\Dc$ is non-degenerate.
\end{theorem}
\begin{proof}
    Suppose $\Dc$ is degenerate. Then there are non-trivial partitions $X = X_1 \cup X_2$, $Y = Y_1 \cup Y_2$ such that
    $\Prob{(x, y) \sim \Dc}{x \in X_1, y \in Y_2} = 0$ and $\Prob{(x, y) \sim \Dc}{x \in X_2, y \in Y_1} = 0$.
    Consider the indicator $I_{Y_1} \in \Fc_Y$ of $Y_1$ that maps $Y_1$ to~$1$ and $Y - Y_1$~--- to $0$.
    Then for every $1 \leq p \leq \infty$ the following equality holds: $\|T_{\Dc} I_{Y_1}\|_{p} = \|I_{Y_1}\|_p$.
    On the other hand by Lemma~\ref{norm-monotonicity}
    \begin{eqnarray*}
        \frac{d}{dp}\|I_{Y_1}\|_p > 0,
        \\
        \frac{d}{dp}\|T_{\Dc} I_{Y_1}\|_p > 0,
    \end{eqnarray*}
    since $\supp \Dc_X = X$, $\supp \Dc_Y = Y$, and $I_{Y_1}$ is non-constant.
    Thus, clearly, $\delta(\Dc) = 0$.
    
    Conversely, let $\Dc$ be non-degenerate.
    Conside the unit sphere $S$ in $L_2(Y)$. For $f \in S$ let us define
    $\delta(f) := \max\setst{\delta \leq 1}{\|T_{\Dc} f\|_{2 + \frac{\delta}{1 - \delta}} \leq \|f\|_{2 - \delta}}$.
    Since $\delta(\Dc) = \inf_{f \in S} \delta(f)$, $\delta(f)$ is continuous, and $S$ is compact, it remains to prove
    that $\delta(f) > 0$ for every $f \in S$.
    If $f$ is constant then, the inequality is obvious since $\|f\|_p$ is constant.
    If $f$ is not constant, then the inequality is true, since $\|T_{\Dc} f\|_2 < \|f\|_2$, so
    there exists $\eta > 0$ such that $\|T_{\Dc} f\|_{2 + \eta} < \|f\|_{2 - \eta}$ ($\|\cdot\|_p$ is continuous in $p$).
\end{proof}

\begin{theorem}
    Let $A \subseteq X$, $B \subseteq Y$. If $\Prob{x \sim \Dc_X}{x \in A} = \mu$ and
    $\Prob{y \sim \Dc_Y}{y \in B} = \nu$, then
    \begin{multline}
        \label{rectangle-bound}
        \Prob{(x, y) \sim \Dc}{x \in A, y \in B} \leq \mu \nu + \\ +
        \left(\mu^{2 - \delta(\Dc)}(1 - \mu) + \mu (1 - \mu)^{2 - \delta(\Dc)}\right)^{1 / (2 - \delta(\Dc))}.
        \left(\nu^{2 - \delta(\Dc)}(1 - \nu) + \nu (1 - \nu)^{2 - \delta(\Dc)}\right)^{1 / (2 - \delta(\Dc))}.
    \end{multline}
\end{theorem}
\begin{proof}
    Let $I_A$ and $I_B$ be indicators of $A$ and $B$ respectively.
    $$
        \Prob{(x, y) \sim \Dc}{x \in A, y \in B}
        =
        \Exp{(x, y) \sim \Dc}{I_A(x) I_B(y)}
    $$
    Let us denote $f(x) := I_A(x) - \mu$, $g(x) := I_B(x) - \nu$.
    $$
        \Exp{(x, y) \sim \Dc}{I_A(x) I_B(y)} = \mu\nu + \Exp{(x, y) \sim \Dc}{f(x) g(y)} = 
        \mu \nu + \Exp{x \sim \Dc_X}{f(x) T_{\Dc} g(x)}
    $$
    By H\"older's inequality
    $$
        \mu \nu + \Exp{x \sim \Dc_X}{f(x) T_{\Dc} g(x)}
        \leq \mu \nu + \|T_{\Dc} f\|_{2 + \frac{\delta(\Dc)}{1 - \delta(\Dc)}} \|g\|_{2 - \delta(\Dc)}.
    $$
    By definition of $\delta(\Dc)$
    \begin{equation}
        \label{holder}
        \mu \nu + \|T_{\Dc} f\|_{2 + \frac{\delta(\Dc)}{1 - \delta(\Dc)}} \|g\|_{2 - \delta(\Dc)}
        \leq \mu \nu + \|f\|_{2 - \delta(\Dc)} \|g\|_{2 - \delta(\Dc)}.
    \end{equation}
    Let us recall that $f = I_A - \mu$, $g = I_B - \nu$.
    Thus,
    \begin{eqnarray}
        \label{norm_f}
        \|f\|_{2 - \delta(\Dc)} = \left(\mu (1 - \mu)^{2 - \delta(\Dc)} +
        \mu^{2 - \delta(\Dc)} (1 - \mu)\right)^{1/(2 - \delta(\Dc))},
        \\
        \label{norm_g}
        \|g\|_{2 - \delta(\Dc)} = \left(\nu (1 - \nu)^{2 - \delta(\Dc)} +
        \nu^{2 - \delta(\Dc)} (1 - \nu)\right)^{1/(2 - \delta(\Dc))}.
    \end{eqnarray}
    Plugging (\ref{norm_f}), (\ref{norm_g}) into (\ref{holder}) we obtain (\ref{rectangle-bound}).
\end{proof}
\end{document}